\documentclass[letterpaper, 10 pt, conference]{ieeeconf}  

\IEEEoverridecommandlockouts    % This command is needed to use the \thanks command
\newtheorem{theorem}{Theorem}

\newtheorem{lemma}{Lemma}

\newtheorem{definition}{Definition}
\newtheorem{proposition}{Proposition}

\newtheorem{assumption}{Assumption}
\let\labelindent\relax
\usepackage{amsmath,amsfonts,amssymb,color}
\usepackage{epsfig}
\usepackage{psfrag}
\usepackage{breqn}
\usepackage{algorithm}
\usepackage{algpseudocode}
\usepackage{array}
\usepackage{epstopdf}
\usepackage{tikz, subcaption, cite}
\usepackage{graphicx}
\usepackage{romannum}
\usepackage{cuted}
\usepackage{scalerel}[2014/03/10]
\usepackage{stackengine}
\usepackage{cleveref}
\usepackage{color}
\usepackage{caption}
\usepackage{subcaption}
\usepackage{mwe}
\usepackage{algpseudocode}
\usepackage{changepage}
\usepackage{enumitem}
\usepackage{fmtcount}
\usetikzlibrary{arrows}
\newcount\Comments  % 1 suppresses notes to selves in text
\newcommand{\kibitz}[2]{\ifnum\Comments=0\textcolor{#1}{#2}\fi}

\title{Leveraging Opinions and Vaccination to Eradicate Networked Epidemics}
\author{Humphrey~C.~H.~Leung,~Zhuocong~Li,~Baike~She,~and~Philip~E.~Par\'e% <-this % stops a space
\thanks{Humphrey~C.~H.~Leung,~Zhuocong~Li,~and~Philip~E.~Par\'e are with the School of Electrical and Computer Engineering at Purdue University. Email: {\tt \{leung61, li3975, philpare\}@purdue.edu}. Baike She is affiliated with the Department of Mechanical and Aerospace Engineering at University of Florida. Email: shebaike@ufl.edu. This work is partially supported by the National Science Foundation, grant NSF-ECCS \#2032258.}
}

\begin{document}
\maketitle
\begin{abstract}                % Vaccines against diseases have proven their effectiveness in saving lives and mitigating epidemic spreading.
We introduce a multi-layer networked compartmental $SIRS-V_o$ model that captures opinion dynamics, disease spread, risk perception, and self-interest vaccine-uptake behavior in an epidemic process.
We characterize the target vaccination criterion of the proposed model and conditions that guarantee the criterion is obtainable by influencing opinions on disease prevalence.
We leverage this result to design an eradication  strategy that leverages opinions and vaccination. 
Through numerical simulations, we show that the proposed eradication strategy is able to stabilize the epidemic process around a healthy state equilibrium, and the outbreak rebounds after the control signal is relaxed.
\end{abstract}

% \begin{keyword}
% Networked Epidemic Modeling; Opinions Dynamics; Multi-layer Networks; Epidemic Control
% %Estimation and control of epidemics, output feedback control, nonlinear observer design, optimal control. 
% \end{keyword}

% \end{frontmatter}

\section{Introduction}

There has been a long history of studying epidemic mitigation and eradication problems \cite{kermack_1927_a}.
% The recent COVID-19 pandemic has drawn more
% %\phil{has drawn}
% % draws 
% research communities, including the control's community, to %epidemic mitigation
% these problems. 
One popular approach for analyzing, predicting, and controlling epidemic spreading processes is to leverage {networked} epidemic spreading models.
% Due to various infection stages and heterogeneous transmission networks, networked compartmental models become one of the most powerful models to carry such research. 
However, 
% modern science and techniques advance researcher's vision on spreading processes, revealing that 
it is also critical to integrate other factors such as human behavior \cite{capasso1978generalization}, opinions \cite{she_2021_peak}, and vaccination \cite{leung2022impact} into the spreading models.
%Hence, in order to study the impact of opinions towards the seriousness of the epidemic, the perceived risk towards the infection, and the hesitation on vaccination (namely, vaccination hesitancy \cite{dub_2013_vaccine}),
Hence, in order to study the impact of vaccination on epidemic mitigation, as determined by opinions on 
% pandemic 
the outbreak's
severity and perceived risk towards the infection, we propose a networked compartmental model by coupling a classic networked susceptible-infected-recovered-susceptible model with 
% vaccination, 
opinion dynamics.
The proposed model also includes vaccination and the loss of immunity, i.e., the vaccinated population can become susceptible again.

Network-structured compartmental models were implemented to study {disease spreading over} heterogeneous mixed population \cite{vanmieghem_2009_virus,mei_2017_on,par_2020_modeling}. 
% For instance, In\cite{vanmieghem_2009_virus}, the author
% % uses Markov process models to 
% described the infection spread over a network through a networked SIS model. 
% Both  \cite{par_2020_modeling} and \cite{vanmieghem_2009_virus} present the stability analysis of networked epidemic models, which is one of the most essential topics in epidemic modeling. In \cite{castellano_2010_thresholds}, the authors study the threshold  behavior of SIS and SIR models. %\phil{threshold} behavior of \phil{networked} SIS and SIR \phil{models.}
% % to find out the threshold and stability of networked models. 
{In order to better capture complex spreading behaviors that are affected by various factors,}
% Besides the classic networked compartmental models (e.g., SIS, SIR, SEIR), 
researchers proposed modern epidemic models by implementing 
% other factors that can affect the spreading process. 
% For example, 
{human}
awareness (opinions)~\cite{weitz_2020_awarenessdriven, she_2021_peak} and vaccination \cite{matthewjameskeeling_2008_modeling,leung2022impact} {to epidemic modeling and control.}
In order to address the critical impact of 
%the above complexity introduced by the human behavior to the problem 
%human behavior
{human awareness}
on vaccination over epidemic spreading networks, mathematical models {from social science} were {imported} to capture these 
%human behaviors 
{opinions}
in epidemic models. 
% from phenomenological and game theoretic perspectives \cite{capasso1978generalization, bauch2004vaccination, bauch2003group}.
% Further, 
% \cite{french_1956_a} and \cite{degroot_1974_reaching} proposed a 
% % the most 
% basic opinion dynamics model -- the French-Degroot model, where \cite{degroot_1974_reaching} defined consensus as 
% % that 
% every node in the network having
% %\phil{having}
% % has 
% the same level of opinion in the final state
% to capture asymptotic behavior of opinions. 
% The papers %\phil{s}
% \cite{proskurnikov_2017_a} and \cite{proskurnikov_2018_a} introduced commonly used opinion models over networks, including Taylor's model, which
An opinion dynamics model was proposed in \cite{Taylor_Towards} which
was 
coupled with epidemic spreading models in
% with opinion dynamics in 
\cite{she_2021_peak}. In addition, multiplex and multi-layer epidemic spreading networks that  incorporate social dynamics and multi-transmission pathways
were studied by \cite{boccaletti2014structure, salehi2015spreading}.
% One of the goals of opinion dynamics analysis is to specify the conditions for opinion convergence and consensus. 

% \par Plenty of work has been done on compartmental models. 
% There are also several mature structures for system dynamics over a network. 
% While many researchers utilize game theory to formulate vaccination with epidemic spread \cite{von2007theory, perisic2009social, li_2017_a}, our work in this paper focuses on the opinion dynamics of believed epidemic severity and people's perceived risk from the neighborhood. 

% We will build a SIRS model, with an additional vaccinated state, incorporating three networks that respectively represent the infection due to physical contact, the belief of epidemic severity represented by opinion dynamics, and perceived infection risk due to the unvaccinated neighbors.

There have been numerous works on game theoretic modeling to study 
% the interaction between infectious disease and volunteer vaccination 
the impact of volunteer vaccination on the spread of infectious disease
over networks \cite{wang2016statistical}.
The basic assumption of these game theoretic models is that people act according to their self interest, maximizing their personal payoff. This assumption has been challenged in more recent works.
Specifically, 
% real data has been used 
simulated data has been used to compare with historical events
to show that imitation dynamics can reproduce salient features of vaccination level over time
% the time evolution of vaccine-uptake 
\cite{fu2011imitation};
and self-interest is a major factor in vaccine-uptake decisions 
% besides self-interest in the real world 
\cite{shim2012influence}.

%\linus{In this work, we explore the way to leverage vaccination to eradicate the epidemic spreading over network. One of novelties is to make the connection between the opinion dynamics and vaccination. Thus, we can increase the vaccinated population proportion by influencing people's opinion so as to suppress the pandemic.}
% In this work, we bridge the gap between epidemic models, opinion dynamics over networks, and behavior modeling, through constructing a novel multi-layer networked epidemic model capturing the aforementioned characteristics.
In this work, we incorporate opinion dynamics and imitation dynamics over networks to construct a novel multi-layer networked epidemic model that captures the aforementioned characteristics of epidemic spreading processes.
% multi-layer networked models, 
% how the networked opinion and individual risk evaluation process influence epidemic spreading processes. 
The main results and contributions of this work are:
\begin{enumerate}
    \item we propose a novel multi-layer network model to connect the opinions towards the disease prevalence, 
    risk perception on the infection, 
    and self-interest vaccination behavior
    % and vaccination 
    on %over 
    networked epidemic processes;
    \item we 
    % analytically 
    show the existence of the healthy state equilibrium and that it is locally stable; %ility condition; 
    \item we characterize conditions on vaccination levels which ensure disease eradication; and
    \item we propose a control strategy for disease eradication through opinions.
\end{enumerate}
While speech censorship, propaganda, and opinion control are characteristics of malfunctioning dystopian states, we hope that through this work we can open up discussions on the role and limitations that opinions play in epidemic control from an analytical and computational perspective.
% To paraphrase a quote about automated systems from the internet, ``People worry that computers will get too smart and take over the world, but the real problem is that they're too [ignorance] and [the computers]'ve already taken over the world''.

We organize the paper as follows: in Section~\ref{sec:model}, we introduce our networked $SIRS-V_o$ model with opinion and imitation dynamics and formulate the problem of interest. 
In Section~\ref{sec:main}, we present the stability analysis of our model and develop a control strategy. 
Section~\ref{sec:sim} demonstrates the numerical results of our work via simulation. 
Section~\ref{sec:conc} concludes the paper and discusses %\phil{es} 
% the 
potential future research directions.%\phil{s}. 

% \vspace{-2ex}

\subsection*{Notation}
\vspace{-1ex}
Let [\(n\)] denote the index set \{1, 2, ..., \(n\)\}, $\forall$ $n 
\in \mathbb{Z}_{>0}$. We view vectors as column vectors and
write $x^{\top}$ to denote the transpose of a column vector $x$. 
We use $x_i$ to denote the $i$th entry of a vector $x$. 
For any matrix $M \in \mathbb{R}^{n\times n}$, we use $[M]_{i,:}$, $[M]_{:,j}$, and $[M]_{ij}$ to denote its $i$th row, $j$th column, and $ij$th entry, respectively. 
Similarly, we denote $[a_{ij}]_{i,j\in [n]}$ as the $n\times n$ matrix for any $n \times n$ 2-dimensional array $a_{ij}$.
We use $\widetilde M$ = $diag \{m_1,...,m_n\}$ to represent a diagonal matrix $\widetilde M \in \mathbb{R}^{n\times n}$ with $[M]_{ii} = m_i,~\forall i \in [n]$. We use $0_n$ and $1_n$ to denote the vectors whose entries all equal 0 and 1, respectively, and $I_n$ to denote the $n\times n$ identity matrix. 
% \hmph{Do we need to change this: }%We define $\tilde{K}[H] = (n-1) \times I_n$.
% \par
For a real square matrix $M$, we use $\alpha(M)$, $\rho (M)$, and $\sigma(M)$ to denote the spectral abscissa (the largest real part among its eigenvalues), spectral radius, and singular value of $M$, respectively.
For any two vectors $v,w \in \mathbb{R}^n$, 
we write $v \geq w$ if $v_i \geq w_i$, 
$v > w$ if $v \geq w$ and $v \neq w$, 
$v \gg w$ when $v_i > w_i$, $\forall i,j \in [n]$.
The comparison notations between vectors are used for matrices as well, for instance, for $A,B \in \mathbb{R}^{n\times n}$, $A > B$ indicates that $A_{ij} > B_{ij}$, $\forall i,j \in [n]$.
% \par
Consider a directed graph $\mathcal{G} = (\mathcal{V}, \mathcal{E})$, with the node set $\mathcal{V} = \{v_1,...,v_n\}$ and the edge set $\mathcal{E} \subseteq \mathcal{V} \times \mathcal{V}$. 
Let matrix $A \in \mathbb{R}^{n\times n},~ [A]_{ij} = a_{ij}$, denote the adjacency matrix of $\mathcal{G} = (\mathcal{V}, \mathcal{E})$, where $a_{ij} \in \mathbb{R}^+$ if $(v_j, v_i) \in \mathcal{E}$ and $a_{ij} = 0$ otherwise, 
% Graph $\mathcal{G}$ does not allow self-loops, i.e., $a_{ii} = 0$, 
$\forall i \in [n]$. 
Let $k_i[A] = \sum_{j \in \mathcal{N}_i}a_{ij}$, %and $d_i[A] = \sum_{j \in \mathcal{N}_i}1$
where $\mathcal{N}_i = \{v_j:(v_i,v_j)\in \mathcal{E}\}$ denotes the neighbor set of $v_i$. %and $|a_{ij}|$ denotes the absolute value of $a_{ij}$. 
The graph Laplacian of $\mathcal{G}$, whose adjacency matrix is $A$, is defined as $L[A] \triangleq \widetilde K[A] - A$, where $\widetilde K[A] \triangleq diag\{k_1[A],...,k_n[A]\}$.

\section{Model And Problem Formulation} \label{sec:model}
% In our previous work \cite{leung_2021_the}, we propose a $SIRS-V_\kappa$ group model. 
%%%
%%%
%%%
In this section, we first introduce a networked ($SIRS-V_o$) model, 
which integrates 
the classic networked susceptible-infected-recovered-susceptible ($SIRS$) model with vaccinations ($V$) and 
opinion ($o$). 
The goal is to incorporate three interacting networks in our compartmental model to capture  the potential impact of media to a disease spreading process through human decision making process.
The three interacting networks include: a disease transition network that captures the disease spreading through physical interactions;
{an} opinion spreading network that captures agents' perception of the current state of the prevalence of the disease; 
and a strategy substitution network that describes the dynamics of agents switching their vaccination strategies to substitute their neighbors' decision of not taking vaccine.
% a strategy imitation in a single comprehensive model.

We first consider an epidemic spreading over a network and use a weighted directed graph $\mathcal{G} = (\mathcal{V}, \mathcal{E})$ to capture the physical transmission network.
% The set 
$ \mathcal{V} =\{v_1,...,v_n\} $ represent{s} {the} set of $n$ nodes and the edge set $\mathcal{E} \subseteq \mathcal{V} \times \mathcal{V}$ represent{s} the disease transmission channels over $\mathcal{V}$.
Each edge is weighted by $\beta_{ij} \in \mathbb{R}_{\geq 0}$, {which} denote{s} the infection rate from host $j$ to $i$. Let $x_i(t) \in [0,1]$, $\forall i \in [n]$, $t \geq t_0$, 
% $x_i(t)$ 
denote the probability that node $i$ is infected at time $t$. 

We define the network of the opinions of the agents toward the prevalence of the epidemic spread
% . 
% We define the opinion spreading network 
as a directed graph $\mathcal{\bar G} = (\mathcal{V}, \mathcal{\bar E})$, where $\mathcal{\bar E} \subseteq \mathcal{V} \times \mathcal{V}$ represents the opinion exchange over $\mathcal{V}$. 
Each edge is weighted by $a_{ij} \in \mathbb{R^{+}}$, denoting the impact of  node $j$ on node $i$ in terms of opinions. Let $o_i(t) \in [0,1]$, $\forall i \in [n]$ at time $t\geq t_0$, denote node $i$'s belief in the seriousness of the epidemic. 
A higher $o_i(t)$ means that node $i$ %is more willing to get vaccinated.
considers the epidemic to be more serious, and vice versa. 
Further, we assume, for node $i$, 
%higher infection level within the community ($x_i(t)$)
that a higher infection level in their neighborhood ($x_j(t)$) will raise their belief in the seriousness of the epidemic. 
% Hence, we adapt Taylor's opinion dynamics model from \cite{proskurnikov_2017_a} for our coupled opinions with epidemics model:
% \begin{equation}\label{eq:Taylor}  
%   \dot {o}_i(t)= \sum_{j \in \mathcal{N}_i} a_{ij}[(o_j(t)-o_i(t))+(x_j(t)-o_i(t))],
% \end{equation}
% where $(o_j(t)-o_i(t))$ is the opinion interchange term, and $(x_j(t)-o_i(t))$ is the prejudice term, which includes the infection levels that the community changes its belief based on.

Lastly, we denote the strategy imitation network as
a weighted directed graph $\mathcal{\hat G} = (\mathcal{V}, \mathcal{\hat E})$, where the edge set $\mathcal{\hat E} \subseteq \mathcal{V} \times \mathcal{V}$ represents
% strategy \hmph{imitation} pathways between agents.
the set of pairs of agents $(i, j)$ where agent $i$ imitates agent $j$'s vaccine-uptake strategy with a non-zero probability.
% the likelihood of a vaccinated node successfully persuade an unprotected node to switch to a vaccinator strategy.
% perceived payoff gain for switching to a vaccinator strategy}
% $\mathcal{V}$. 
% In imitation dynamic, we assumed that agents randomly sample other agents according to some distributions, if the strategy of the sampled agent provides a higher payoff, then the agent will switch to the new strategy with a probability proportional to the expected gain in payoff \cite{bauch2005imitation, hofbauer1998evolutionary}.
% In semi-anonymous graphical game with strategic substitution, a player’s incentives to take an action are decreasing in the number of his or her friends who take the action \cite{jackson2015games}.
% Inspired by imitation dynamic \cite{bauch2005imitation, hofbauer1998evolutionary} and semi-anonymous graphical game with strategic substitution \cite{jackson2015games},
Furthermore, we assume in this model that each agent $i$ randomly samples another agent $j$ according to some distribution $p_{ij}$; 
% if the strategy of the sampled agent provides a lower payoff compared to any alternative strategies (vaccine-uptake decision in this case), 
% then the altruist will switch to the alternative strategy on its neighbor's behalf with a probability proportional to its neighbor's expected payoff gain.
if the strategy of the sampled agent provides a higher payoff, then the agent will switch to the new strategy with a probability proportional to the expected payoff gain \cite{bauch2005imitation, hofbauer1998evolutionary}.

The weight of each directed edge $(v_j,v_i)$ is $\eta_{ij}(t)$, and it is defined in the following way.
Let $u^v_i$ be the probability of a serious vaccine side effects, such as allergic reactions,
after vaccine-uptake; 
$u^x_i$ be the probability of significant morbidity after an infection of agent $i$; and 
$p_{ij}$ be the probability of agent $i$'s decision being influenced by the decision of agent $j$
% , which can also be interpreted as the relative importance of agent $j$ to $i$, 
and $[p_{ij}]_{i,j\in[n]}$ is a row stochastic matrix.
Let
$m^v_i$ and $m^x_i$ represent the severity of a serious vaccine side effects and significant morbidity after infection, respectively,
and $c$ be the government compensation for serious vaccine side effects such that $c > m^v_i u^v_i$ for all $i \in [n]$.
Then the payoff of getting vaccinated is $c-m^v_i u^v_i$,
and the \textit{perceived payoff} of $i$ not getting vaccinated is $-o_i(t)(m^x_i u^x_i)$, 
where $o_i(t)$ is the perceived disease prevalence of node $i$ at time $t$.
Note that $c$ can be interpreted as the sum of the expected payoff of vaccine-uptake and expected cost of remaining unprotected that is not captured by $m^x_i u^x_i$ and $m^v_i u^v_i$, such as the inconvenience of mask wearing. 
The above interpretation only suggest one scenario to interpret these parameters. 
The above \textit{non-vaccine-uptake} payoff formulation assumes that each agent %bases 
determines
its payoff evaluation on the current perceived disease prevalence.
We further assume that the vaccine has perfect efficacy; 
an agent takes the vaccine once it switches its strategy to \textit{vaccine-uptake};  
and the agent cannot switch back except through the loss of immunity.
Therefore, the \textit{perceived payoff gain} for agent $i$ 
% to substitute $j$ 
% to take the vaccine to provide protection of $j$ 
% to take the vaccine 
at time $t$ is $g_{i}(t) = o_i(t)(m^x_i u^x_i) + (c - m^v_i u^v_i)$.
% \hmph{which means that agent $i$ receives the aggregated payoff gains of its neighbors when agent $i$ chooses to take the vaccine on their behalf.}

Finally, let $\mu_i$ be the sampling rate of agent $i$ and $q$ be an arbitrary constant. 
Then,
we define the \textit{strategy transition matrix} 
% as a $n \times n$ matrix 
$H(o(t)) \in \mathbb{R}^{n \times n}$ such that $[H(o(t))]_{ij} = \eta_{ij}(t) = q\mu_i g_i(t)p_{ij}(t)$.
We denote 
$\eta^{min}_{ij} = q\mu_i(c - m_i^v u_i^v) p_{ij}$, 
$\Delta \eta_{ij} = q\mu_i(m^x_i u^x_i) p_{ij}$, 
and $\eta_{ij}(t) = o_i(t)\Delta \eta_{ij} + \eta^{min}_{ij}$.
We now formally introduce the networked ($SIRS-V_o$) model:

\vspace{-6ex}

\begingroup\makeatletter\def\f@size{8}\check@mathfonts
\begin{subequations}\label{equ:tot}
\begin{align}
\dot {o}_i(t)&= \sum_{j \in \mathcal{N}_i}  a_{ij}(o_j(t)-o_i(t))+ a_{ij}(x_j(t)-o_i(t)),\label{equ:o}\\
\dot {s}_i(t)&= \delta_iv_i(t) - s_i(t)\sum_{j \in \mathcal{N}_i}[\beta_{ij}x_j(t) + \eta_{ij}(t)v_j(t)]+\omega_i r_i(t),\label{equ:s}\\
\dot {x}_i(t)&= s_i(t)\sum_{j \in \mathcal{N}_i}\beta_{ij}x_j(t) - \gamma_ix_i(t),\label{equ:x}\\
\dot {r}_i(t)&= \gamma_ix_i(t) - \omega_i r_i(t)  -r_i(t)\sum_{j \in \mathcal{N}_i}\eta_{ij}(t) v_j(t),\label{equ:r}\\
\dot {v}_i(t)&=(s_i(t)+r_i(t)) \sum_{j \in \mathcal{N}_i}\eta_{ij}(t)v_j(t) -\delta_i v_i(t),\label{equ:v}
\end{align}
\end{subequations}
\endgroup

\vspace{-2ex}

\noindent
% \par We assume all communities have a minimum perceived risk from their neighbors, denoted by $\eta_{min}$. 
% When the community considers the disease most serious, i.e., $o_i(t)=1$, it perceives the maximum risk of $\eta_{ij}$ from its neighborhood and vice versa. 
% Therefore, we define the networked ($SIRS-V_o$) model as \eqref{equ:tot}. 
where $s_i$, $x_i$, $r_i$, and $v_i$ represent the percentage of susceptible, infected, removed, and vaccinated populations at node $i$.
Note that $\eta_{ij}(t) = ((1-o_i(t))\eta_{ij}+o_i(t)\eta^{\min}_{ij})$ is a convex combination of $\eta_{ij}$ and $\eta^{\min}_{ij}$.
% where $0\leq \eta_{\min} \leq \min\{\eta_{ij}: \eta_{ij} > 0\}$ and $\eta^{\min}_{ij} \in \mathbb{R}^{0+}$.
% Furthermore, the above formulation will be identical to its selfish imitation driven vaccination dynamic counterpart if we replace $(1 - v_j(t))$ in \eqref{equ:s}, \eqref{equ:r}, and \eqref{equ:v} {with $v_j(t)$}, and replace $m^x_j u^x_j$ in the perceived payoff gain $f_{ij}(t)$ with $m^x_i u^x_i$.
To further simplify the model, we express \eqref{equ:tot} in a compact form:
% 
% \vspace{-4ex}
% 
% \\
% \begin{strip}
% \begingroup\makeatletter\def\f@size{8}\check@mathfonts
\begingroup\makeatletter\def\f@size{9}\check@mathfonts
\begin{subequations}
\label{mat:tot}
\begin{align}
    \dot{o}(t) &= A(x(t)-o(t))-2L[A]o(t),\label{mat:o}\\
    \dot{s}(t) &= \widetilde Dv(t) -\widetilde{S}(t)(Bx(t) + H(o(t))v(t) + \widetilde Wr(t) 
    % \nonumber\\
    % & \ \ \ \ \ -\widetilde{S}(t)H(o(t))v(t)
    ,\label{mat:s}\\
    \dot x(t) &= \widetilde{S}(t)Bx(t) - \widetilde Gx(t),\label{mat:x}\\
    \dot r(t) &= \widetilde Gx(t) - \widetilde Wr(t)-\widetilde{R}(t)H(o(t))v(t),\label{mat:r}\\
    \dot v(t) &= (\widetilde{S}(t)+\widetilde{R}(t)) H(o(t))v(t)- \widetilde Dv(t),\label{mat:v}
\end{align}
\end{subequations}
\endgroup
% 
% \endgroup
% \end{strip}
where $\widetilde{S}(t) = diag(s(t))$, $\widetilde{R}(t) = diag(r(t))$, and $L[A]$ is the Laplacian matrix of the opinion spreading graph $\mathcal{\bar{G}}$. 
% Matrices 
$\widetilde G$  and $\widetilde D$ are diagonal matrices, with $[\widetilde{G}]_{ii} = \gamma_i$ and $[\widetilde{D}]_{ii} = \delta_i$, $\forall i \in [n]$. 
We define $H(o(t))$ as a function of $o(t)$, where $H(o(t)) = \widetilde{O}(t)\Delta H + H_{\min}$, 
and 
% $H_{\min} \in \{0, 1\}^{n\times n}$ is a binary matrix such that 
$[H_{\min}]_{ij} = \eta^{\min}_{ij}$ if $\eta_{ij} > 0$, otherwise, $[H_{\min}]_{ij} = 0$.
%where $[H(o(t))]_{ij} =(o_i(t)\eta_{ij}+(1-o_i(t))\eta_{min})$, $\forall i,j \in [n]$
% \m

% \par {After defining the model, we state the problem we are motivated to explore in this work. 
% First, based on the defined model in \eqref{equ:tot}, we are interested in the steady-state behavior of the model. 
% We wonder if the networked $SIRS-V_{o}$ model exhibits different property in terms its equilibria. 
% Based on the model formulation where the infections, opinions, and perceived risk against vaccination, 
% we are interested in the interactions between the infected population, the opinions, and the perceived risk against vaccination. 
% Last, based on the potential connections, we will explore whether we can change the opinions over the network to tackle epidemic mitigation problems.}

Recall that our goal is to investigate the potential impact of media on epidemic spreading through human decision processes, we introduce the media actuator $m(t) = [m_i(t)]_{i=1\hdots n}$ and augment \eqref{mat:o} by including the set of media nodes:
\begin{equation} \label{eq:media_input}
    \dot o(t) = A(x(t) - o(t)) - 2L[A]o(t) + \widetilde M(t)(1 - o(t)),
\end{equation}
where $\widetilde M(t) = diag(M(t))$.
The problems that we will answer in this paper include:
\begin{enumerate}
% [label={P\arabic*.}]
    % \item \cut{\label{prob:consensus} What is the form of consensus equilibrium of \eqref{mat:o} without the control of $g(t)$, which will be discussed in Section \ref{sec:opinion} 
    % }
    % \item \label{prob:epid_eq} The existence and uniqueness of a healthy state equilibrium of system \eqref{mat:tot} and the local stability condition of its infection subsystem, which will be discussed in Section \ref{sec:epid}
    \item \label{prob:epid_eq} %How to s
    Show the existence and uniqueness of a healthy state equilibrium of system \eqref{mat:tot} and the local stability condition of its infection subsystem.
    % \item \label{prob:control} The existence of an eradication algorithm through the control of a media actuator $g(t)$ and the sufficient condition for such algorithm to operate, which will be discussed in Section \ref{sec:control}.
    \item \label{prob:control} Construct an eradication algorithm stabilizing \ref{mat:x} around $x = 0_n$ through the control of a media actuator $g(t)$, and provide
    %? Furthermore, what are 
    the assumptions and conditions that guarantee the validity of the algorithm.
\end{enumerate}
Solutions to these problems will be provided 
% answered 
in Section~\ref{sec:epid} and Section~\ref{sec:control}, respectively.
\section{Main Results}\label{sec:main}
This section explores the stability conditions, the properties of the opinion dynamics, and the strategies to mitigate the networked epidemic process
% for the system 
in \eqref{mat:tot}. 
% We construct the compact form of our model to simplify the analysis. 
%The theorems are for the general conditions, while the corollaries focus on the conditions for healthy state equilibria.
% We write \eqref{equ:tot} in a compact form:\\
% % \begin{strip}
% \begingroup\makeatletter\def\f@size{8}\check@mathfonts
% \begin{subequations}\label{mat:tot}\begin{align}
%     \dot o(t) &= A(X(t)+V(t)-o(t))-2L[A]o(t),\label{mat:o}\\
%     \dot s(t) &= -\tilde{S}(t)Bx(t) + Wr(t)-\tilde{S}(t)P(d - H(o(t))v(t)) + Dv(t),\label{mat:s}\\
%     \dot x(t) &= \tilde{S}(t)Bx(t) - Gx(t),\label{mat:x}\\
%     \dot r(t) &= Gx(t) - Wr(t)-\tilde{R}(t)P(d - H(o(t))v(t)),\label{mat:r}\\
%     \dot v(t) &= (\tilde{S}(t)+\tilde{R}(t))PH(o(t))(1_n-v(t))- Dv(t),\label{mat:v}
% \end{align}\end{subequations}
% \endgroup
% % \end{strip}
% where $\tilde{S}(t) = diag(S(t))$, $\tilde{R}(t) = diag(R(t))$, $\tilde{K}[H] = (n-1)\times 1_n$, and $L[A]$ is the laplacian matrix of the opinion spreading graph $\mathcal{\bar{G}}$. Matrices G, P, and D are diagonal matrices, with $[G]_{ii} = \gamma_i$, $[P]_{ii} = \rho_i$, and $[D]_{ii} = \delta_i$, $\forall i \in [n]$. Define $H(o(t))$ as a function of $O(t)$, where $[H(o(t))]_{ij} =(1-o_i(t))\eta_{ij}+o_i(t)\eta_{min}$, $\forall i \in [n]$ and $j \in \mathcal{N}_i$. 
%Therefore, we have the following results:
The results are developed under the following assumptions.
%\begin{assumption}
%\label{ass:v0}
%We assume the dot product $\left<[H]_{i,:},~ 1_n\right>\leq (n-1), ~\forall i \in [n]$. Thus, when the vaccine is perfect %(i.e., $\delta_i = 0$), $v_i(t)$ is monotonically increasing.
%\end{assumption}
%\hmph{Is this still true?}

\begin{assumption}
\label{ass:non-neg}
% Considering the real-world meaning of parameters, 
All parameters are real and non-negative.
% The immunity waning rate $\delta_i > 0$ for all $i \in [n]$.
\end{assumption}

\begin{assumption}
\label{ass:ini}
%\phil{[add initial condition assumption]}
The initial conditions at $t_0$ should obey $(o_i(t_0),s_i(t_0),x_i(t_0),r_i(t_0),v_i(t_0)) \in [0,1]$ and $s_i(t_0)+x_i(t_0)+r_i(t_0)+v_i(t_0) = 1,~\forall i \in [n]$.
\end{assumption}

\begin{assumption}\label{ass:op_graph_strong_connect}
%We do not allow self-loops in any of %\phil{the}
%the three graphs. Moreover, 
% We do not allow self-loops in the opinion spreading graph $\bar {\mathcal{G}}$, and 
We assume that the opinion spreading graph $\bar {\mathcal{G}}$, disease transmission graph ${\mathcal{G}}$, and strategy imitation graph $\mathcal{\hat G}$ are strongly connected.
\end{assumption}
%\hmph{Is this still true?}

First, one can show that the model in \eqref{mat:tot} is well defined by checking the derivative of a state being always pointing inside the set $[0,1]$, when the state hits the boundary of $[0,1]$. Thus, we omit the proof by giving the following statements.
% we show the model in \eqref{mat:tot} is well-defined.
\begin{lemma}\label{lemma:const}\label{lemma:pos}
%The sum of 4 compartments are constant, i.e., %$s_i(t)+x_i(t)+r_i(t)+v_i(t)=1$, $\forall i \in [n]$ and  $t > t_0.$
% At any time
For all $t>t_0$ and $i \in [n]$, $s_i(t)+x_i(t)+r_i(t)+v_i(t)=1$ and ($s_i(t)$, $x_i(t)$, $r_i(t)$, $v_i(t)$) $\in [0,1]$. %, $\forall i \in [n]$.
\end{lemma}
% \begin{proof}
% We first observe $\int_{t_o}^{t}\dot s_i(\tau)+\dot x_i(\tau)+\dot r_i(\tau)+\dot v_i(\tau)d\tau=0$ and $s_i(t_0)+x_i(t_0)+r_i(t_0)+v_i(t_0) = 1$.
% Therefore, $s_i(t)+x_i(t)+r_i(t)+v_i(t)=1$, $\forall i \in [n]~ and~ t \geq t_0$.
% \par When $v_i(t) \uparrow 1$, $s_i(t) = r_i(t)=0$. Then, $\lim_{v_i(t) \uparrow 1}\dot v_i(t) = -\delta_i v_i(t)$, which is non-positive. When $v_i(t) \downarrow 0$, $\lim_{v_i(t) \downarrow 0}\dot v_i(t) = (s_i(t)+r_i(t))\rho_i \sum_{j \in N_i} \eta_{ij}(t)$, which is non-negative. Therefore, $v_i(t) \in [0,1]$, $\forall i \in [n]$ and $t > t_0$. Similarly, we can prove that $(s_i(t),x_i(t),r_i(t)) \in [0,1],$ $\forall i \in [n]$ and $t>t_0$.
% \end{proof} 
\begin{lemma}
\label{lem:op}
If $o(t_0) \in [0,1]^n$, $o_i(t)$ varies in the range of [0,1], $\forall i \in [n]$ and $t>t_0$.
\end{lemma}
% \begin{proof}
%  According to Lemma \ref{lemma:const}, $x_j(t)\in [0,1]$ and $o_j(t)\in [0,1]$. %It has $\lim_{o_i(t) \downarrow 0}\dot o(t)\geq 0$.
%  When $o_i(t) \uparrow 1$,
% \begin{equation*}
% \dot {o}_i(t)= \sum_{j \in \mathcal{N}_i} a_{ij}(o_j(t)-1)+ a_{ij}(x_j(t)-1),    
% \end{equation*}
% where $ a_{ij} \geq 0$, $\forall i, j \in[n]$.
% Then, it has $\lim_{o_i(t) \uparrow 1}\dot o(t)\leq 0$. Similarly, we can prove $\lim_{o_i(t) \downarrow 0}\dot o(t)\geq 0$. 
% Therefore, $o_i(t) \in [0,1]$, $\forall i \in [n]$ and $t\geq t_0$.
% \end{proof}
% 
% {After discussing that the states of the model are in the range of $[0,1]$, 
% we now explain the behavior of the vaccinated compartment.}
\begin{proposition}
\label{prop:1}
Without the loss of vaccine immunity, the vaccinated population monotonically increases. 
% The vaccination process can never make people leave the vaccinated state.
\end{proposition}
% \begin{proof}
% {Consider the vaccinated proportion in community $i$.}
% According to Lemma \ref{lemma:pos}, ($s_i(t)$, $r_i(t)$, $v_i(t)$) $\in [0,1]$, $\forall i \in [n]$ and $t \geq t_0$. Thus, $(1-v_j(t)) > 0$, $\forall j \in \mathcal{N}_i$, $i \in [n]$, and $t \geq t_0$. Recall that $\rho_i > 0$ and $\eta_{ij}(o(t)) \geq 0$, $\forall i,j \in [n]$. Therefore, $\dot v_i(t) \geq 0$, $\forall i \in [n]$ and $t \geq t_0$, if $\delta = 0_n$.
% \end{proof}
% \par {Based on Lemma~\ref{lemma:const}, Lemma~\ref{lem:op}, and Proposition~\ref{prop:1}, the model in \eqref{mat:tot} is well-defined.} 
Let $( o^*, s^*, x^*, r^*, v^*)$ denote the equilibria of the $SIRS-V_o$  model.
We start to analyze the equilibria of the $SIRS-V_o$ model through the following definition.
% We define two types of equilibria in our model as follows:
% \begin{definition}(Healthy State Equilibrium). \emph{A healthy state} equilibrium is an equilibrium with the steady-state vector of infected populations $x^*=0_n$.
% % , where $\dot s^* = \dot x^* = \dot r^* = \dot v^* = \dot o^* =0_n$.}
% \end{definition}
% \begin{definition}(Healthy State). \emph{A healthy state} is 
% where there is no infected population, i.e., $x(t) = 0_n$.
% \end{definition}
\begin{definition}(Healthy State Equilibrium). \emph{A healthy state equilibrium} is 
an equilibrium with the steady-state vector of infected populations $x^*=0_n$
% where there is no infected population, i.e., $x(t) = 0_n$.
, where $\dot s^* = \dot x^* = \dot r^* = \dot v^* = \dot o^* =0_n$.
\end{definition}

\subsection{Healthy State Equilibrium}\label{sec:epid}
%\par With the results above, we can study the stability conditions of the equilibria.
% In Section~\ref{sec:opinion}, we study the properties of the coupled opinion dynamics in \eqref{mat:o}.  
%In Section \Romannum{2}.B, 
In this subsection,
we %will 
% leverage the results developed in Section~\ref{sec:opinion} to
% investigate the %stability conditions of the equilibria.
prove that the healthy state equilibrium of the coupled networked system in \eqref{mat:tot} exists
and derive 
% its basic and effective reproduction numbers. %,
the stability condition of the infection subsystem around any healthy state equilibrium.
% and that the healthy state equilibrium is unique, 
% i.e., when $x(t) = 0_n$ at an equilibrium, the other states have only one possible value.
\begin{lemma}\label{lemma:healthy}
At any healthy state equilibrium, where $x^*=0_n$, the recovered population $r^* = 0_n$. Furthermore, 
% it implies 
$s^*+v^*=1_n$.
\end{lemma}
\begin{proof}
According to \eqref{mat:r}, when $x^* = 0_n$, $$\widetilde Wr^*+\widetilde{R}^*H_{min} v^*=0_n,$$ where $[\widetilde W]_{ii} > 0$, $v^* \geq 0_n$, and $H_{min} \geq 0_{n \times n}$, $\forall i \in [n]$. Therefore, $r^* = 0_n$, and $s^*+v^*=1_n$ via Lemma~\ref{lemma:const}.
\end{proof}
\begin{lemma}\label{lem:healthycons}
% At a healthy state equilibrium,
At any healthy state equilibrium, where $x^*=0_n$,
%no matter what is the initial condition, 
the opinions will reach consensus at $o_i^* = 0$, $\forall i \in [n]$.
\end{lemma}
\begin{proof}
At an equilibrium where $x^*=0_n$, the opinions in \eqref{mat:o} satisfy
\begin{align*}
    0_n &= A(0_n-o^*)-2L[A]o^*\\
    &=-(L[A]+\widetilde{K}[A])o^*.
\end{align*}
Based on the Gershgorin circle theorem, we have that the spectrum of the matrix lies within at least one of the Gershgorin discs, where the center of the $i$th disc is given by $c_{ii}=[-(L[A]+\widetilde{K}[A])]_{ii}=-2k_i[A]$, and the corresponding  radius of the disc is $[L[A]]_{ii}=k_i[A]$, $\forall i \in [n]$. Hence, all discs are strictly located on the left-hand side of the complex plane, which indicates that the matrix $-(L[A]+\widetilde{K}[A])$ is negative definite, i.e., the null space of the matrix contains only $0_n$. Further, at an equilibirium where $x^*=0_n$, the only solution for \eqref{mat:o} is $o^*=0_n$, which completes the proof. %\qed
\end{proof}
%  \begin{equation*}
%     [L[A]+\widetilde{K}[A])]_{ii}
% \end{equation*}
% \end{equation*}
% Note that 

% To show that the opinions will reach consensus at $o_i^* = 0$ at a healthy state equilibrium, 
% we first show that the system in \eqref{mat:o} reaches consensus at 
% $o^*$ if and only if
% $x_i^* = x_j^*$, $\forall i,j \in [n]$. 
% Moreover, when reaching consensus, we show that $o^*_i = x^*_i$, $\forall i\in [n]$.

% Suppose $x_i^* = x_j^*$, $\forall i,j \in [n]$ holds, the system reaches a consensus where $o_j^* = o_i^*$, $\forall i,j \in [n]$, according to \eqref{equ:o}. 
% Suppose $o_j^* = o_i^*$, $\forall i,j \in [n]$, holds, then $x_j^* = o_i^*$, $\forall i,j \in [n]$, so as to make $\dot o(t) = 0_n$. 
% Therefore, the necessary and sufficient condition for a consensus is $x_i^* = x_j^* (= o_i^*)$, $\forall i,j \in [n]$.
% {Further, when reaching consensus, i.e., $o^*_i = o^*_j$, based on \eqref{equ:o}, we must have $x^*_j=o^*_i$ and $x^*_i=o^*_j$ $\forall i,j \in [n]$. Hence, we have $o^*_i = x_i$.}
% Therefore, when $x^* = 0_n$, $o* = 0_n$ at the healthy state equilibrium.
% % When $x^* = 0_n$, $0_n = \dot o^* = -(\widetilde K[A] + L[A])o^*$, which implies that $\widetilde K[A]o^* \phil{= 0_n}$ \phil{[$=-L[A]o^*$?]}. By Assumption~\ref{ass:op_graph_strong_connect}, $\widetilde K[A]_{ii} > 0$ for all $i \in [n]$, which implies $o^* = 0_n$.
% \qed\end{proof}
The above lemmas outline some important characteristics about the healthy state equilibrium when it exists, which will be useful when proving the following proposition.
\begin{proposition}\label{Th:healthy_state_ex_n_uniq}
If $\alpha(-D + H_{min}) \leq 0$, then $(0_n, 1_n, 0_n, 0_n, 0_n)$  is the unique healthy state equilibrium of system~\eqref{mat:tot}.
If $\alpha(-D + H_{min}) > 0$, then system~\eqref{mat:tot} has two healthy state equilibria, $(0_n, 1_n, 0_n, 0_n, 0_n)$ and $(0_n, 1_n - v^*, 0_n, 0_n, v^*)$, where $v^* \gg 0_n$.
\end{proposition}
\begin{proof}
By Lemma~\ref{lemma:const}, \ref{lemma:healthy}, and \ref{lem:healthycons}, 
if there exists a healthy state equilibrium, 
then it will have the form $(0_n, 1_n - v^*, 0_n, 0_n, v^*)$.
Furthermore, system~\eqref{mat:tot} can be simplified to: 
\begin{equation}\label{eq:healthy_sdot}
    \dot{v}(t) = (I - \widetilde{V}(t))H_{\min}v(t) - \widetilde{D}v(t),
\end{equation}
where $H(o(t)) = H_{\min}$ since $o(t) = 0_n$.
% within the domain of interest.
Therefore, to show that a healthy state equilibrium exists,
we limit our search of the healthy state equilibrium within the domain of $\{0_n\} \times [0, 1]^n \times \{0_n\} \times \{0_n\} \times [0, 1]^n$ through Brouwer's fixed-point theorem \cite[pg.~20]{brown1993topological}.
After the above simplifications, the problem of showing the existence and uniqueness of the healthy state equilibrium of system~\eqref{mat:tot} is reduced to the equilibrium analysis of the single viral networked SIS model, 
% \begin{equation}\label{eq:sis}
%     \dot{x}(t) = (I - \widetilde{X}(t))Bx(t) - \widetilde{D}x(t),
% \end{equation}
by replacing $v$ with $x$ and $H_\min$ with $B$.
The desired result then follows directly from \cite[Prop.~2]{liu2019analysis} since $H_\min$ is irreducible by Assumption~\ref{ass:op_graph_strong_connect}.
% and simplify system \eqref{mat:tot} to:
% \begin{equation}\label{eq:healthy_sdot}
%     \dot{v}(t) = (I - \widetilde{V}(t))H(o(t))v(t) - \widetilde{D}v(t),
% \end{equation}
% where $H(o(t)) = H_{\min}$ since $o(t) = 0_n$ within the domain of interest.
\end{proof}
Proposition~\ref{Th:healthy_state_ex_n_uniq} implies that 
% at 
the networked model in 
\eqref{mat:tot} has a unique healthy equilibrium when $\alpha(-D + H_\min) < 0$, where the susceptible and vaccinated population within the community will reach a balance at a unique ratio, which is determined by the model parameters.

\begin{lemma}\label{lem:x_rep_num}
% Suppose that $\widetilde{S}^* B$ is irreducible. 
The infection subsystem \eqref{mat:x} is locally asymptotically stable around $x^* = 0$, if 
$\rho(\widetilde G^{-1}\widetilde{S}^* B) < 1$.
\end{lemma}
\begin{proof}
    Considering the Jacobian matrix of \eqref{mat:x}, we have that if the spectral abscissa of $\widetilde S(t)B - \widetilde G$ is less than $0$, then \eqref{mat:x} is locally asymptotically stable around $x^* = 0$.
    Note that $\widetilde{S}^*$ is a positive diagonal matrix and $B$ is irreducible by Assumption~\ref{ass:op_graph_strong_connect}.
    The result then follows as a direct application of \cite[Prop. 1]{liu2019analysis} to~\eqref{mat:x}.
    % the construction of the next generation matrix method detailed in \cite{diekmann2010construction}, and in the same paper \cite[Lem A.4]{diekmann2010construction} it is shown that $sign(1 - \widetilde{S}^* B\widetilde G^{-1}) =  sign(\widetilde{S}^* B-\widetilde G)$ if $\widetilde{S}^* B-\widetilde G$ is irreducible.
% \qed
\end{proof}
Since $S^* \leq 1_n$, $\rho(\widetilde G^{-1}B) < 1$ implies $\rho(\widetilde G^{-1}\widetilde{S}^* B) < 1$ and in turn implies local asymptotic stability of the infection subsystem \eqref{mat:x}. We define $\mathcal{R}_0 = \rho(\widetilde {G}^{-1}B)$ as the \textit{basic reproduction number} and $\mathcal{R}_t = \rho(\widetilde G^{-1}\widetilde{S}(t) B)$ as the \textit{effective reproduction number} of the subsystem \eqref{mat:x}.
\subsection{Control Strategy}\label{sec:control}
% % % % % %
% % % % % %
% Next, we consider a control strategy for epidemic eradication. 
In this section, we consider a fixed opinion media actuator, which connects to each node $i$ in the opinion network with strength $m_i(t) \in [0, \infty)$.
Equation \eqref{equ:o} can be rewritten as:
\begingroup\makeatletter\def\f@size{8}\check@mathfonts
\begin{equation*}
    \dot {o}_i(t) = m_i(t)(1 - o_i(t)) + \sum_{j \in \mathcal{N}_i} a_{ij}(o_j(t)-o_i(t))+ a_{ij}(x_j(t)-o_i(t))
\end{equation*}
\endgroup
to include the influence of the actuator, and the corresponding matrix form of the system with the actuator input is:
\begingroup\makeatletter\def\f@size{9}\check@mathfonts
% \begin{subequations}
\begin{equation}
    \dot{o}(t) = A(x(t)-o(t))-2L[A]o(t) + \widetilde{M}(t)(1_n - o(t)),\label{mat:o_w_input}
    % \\
    % \dot{s}(t) &= \widetilde Dv(t) -\widetilde{S}(t)(Bx(t) + H(o(t))v(t) + \widetilde Wr(t) 
    % % \nonumber\\
    % % & \ \ \ \ \ -\widetilde{S}(t)H(o(t))v(t)
    % ,\label{mat:s_w_input}\\
    % \dot x(t) &= \widetilde{S}(t)Bx(t) - \widetilde Gx(t),\label{mat:x_w_input}\\
    % \dot r(t) &= \widetilde Gx(t) - \widetilde Wr(t)-\widetilde{R}(t)H(o(t))v(t),\label{mat:r_w_input}\\
    % \dot v(t) &= (\widetilde{S}(t)+\widetilde{R}(t)) H(o(t))v(t)- \widetilde Dv(t),\label{mat:v_w_input}
\end{equation}
\endgroup
where $\widetilde{M}(t) = diag(m(t))$.
% We further assume that $x(t)$ is known to the state feedback controller $m(t)$.
% , that allows us to write $m(t) = m(x(t))$.
% We will show later in this section that information about infected states $x(t)$ is not necessary if we allow for a more conservative controller.

Intuitively, if $m_i(t)$ are maintained on a sufficiently high level, then the epidemic will eventually die out because individuals will be motivated to take the vaccination due to their perception of the high disease prevalence. 
The main goal of this section is to derive a lower bound on $m(t)$ that is sufficient to guarantee the local stability of the infection subsystem \eqref{mat:x} around the healthy state equilibrium.
To achieve this goal, we will introduce two concepts which we will formally characterize for system~\eqref{mat:tot} with the media actuator in \eqref{mat:o_w_input} later in the section.
They are the \textit{target vaccination criterion} and the \textit{target opinion criterion}:
\begin{definition}[Target Vaccination Criterion]\label{def:target_vaccination_criterion}
    We call $v_c$ %as 
    a target vaccination criterion, 
    if for every $v(t)$ that satisfies the target vaccination criterion $v_c$, that is, $v_i(t) > v_c$ for all $i\in [n]$ and $t \geq T$ for some $T \in \mathbb{R}_{\geq 0}$,
    % for every $v(t)$ such that 
    % then 
    % the corresponding infection subsystem \eqref{mat:x}
    the healthy state equilibrium is locally asymptotically stable 
    for the corresponding infection subsystem \eqref{mat:x}.
    % around the healthy state equilibrium.}
    % for every $v(t)$ such that $v_i(t) > v_c$ for all $i\in [n]$ and $t \geq T$ for some $T \in \mathbb{R}_{\geq 0}$.
    % such that:
    %     $$v_i(t) > v_c$$ 
    % for all $t \geq T$ and $i\in [n]$.
\end{definition}
In other words, if the vaccination levels of all $n$ nodes 
% $i \in [n]$ 
stay above the target vaccination criterion, $v_c$, when time $t$ is sufficiently large, then the epidemic will die out.
% \hmph{
Furthermore, since the eradication solution is trivial if $v_c \leq 0$ and infeasible if $v_c = 1$ 
% due to the loss of immunity
, we only consider the case where $v_c \in (0, 1)$.
% % % % %
\begin{definition}[Target Opinion Criterion]\label{def:target_opinion_criterion}
    We define $\underline o_i^{v_c}$ as a target opinion criterion with respect to $v_c$
    % at $x_i(t) = x_i$
    if 
    % there exists a $T \in \mathbb{R}_{\geq 0}$ such that 
    $v_i(t) > v_c$
    % for all $i \in [n]$
    % and $t \geq T$
    when $o_i(t) \geq \underline o_i^{v_c}$ 
    for all $i \in [n]$ and $t \geq T$, where $T \geq 0$.
    % as $t \to \infty$.
\end{definition}
% Unlike the critical vaccination threshold which apply to all node $i \in [n]$, the critical opinion threshold is a per node condition applied to a specific value of $x_i(t)$.
The target opinion criterion is a lower bound on the opinion that guarantees the vaccination level at node $i$ to remain above the target vaccination criterion.
The control criterion $\underline m^{v_c}$ will be expressed as a lower bound of $m(t)$ in terms of the target opinion criterion $\underline o_i^{v_c}$ such that the vaccination level at each node $v_i(t)$ remains above $v_c$ despite loss of immunity. 
% Notice that $\underline o^{v_c}, \underline m^{v_c}$ are not functions of time, but functions of scalar values such that $\underline o: (0, 1) \to (0, 1)^n$, $\underline m: (0, 1) \to [0, \infty)^n$. 
Before introducing the key results, we make the following assumption for the control strategy. 
% \hmph{
% We denote the time instance when the controller starts applying to the system as $t_b$, the time instance when all $o_i(t)$ crosses the threshold $o^*_i(v_c, x_i(t))$ as $t_k$, the time instance when all $v_i(t)$ crosses the threshold $v_c$ as $t_c$.
\begin{assumption}\label{asm:symm_B}
    The infection transmission matrix $B$ is symmetric, i.e. $B = B^\top$.
\end{assumption}
The following theorem characterizes the target vaccination criterion of \eqref{mat:tot}.
\begin{theorem}
\label{thm:sufficient_x_stable}
    Under Assumptions~\ref{ass:op_graph_strong_connect} and \ref{asm:symm_B},
    % a target vaccination criterion of \eqref{mat:tot} is:
    % if 
    \begin{equation}\label{eq:target_vaccination_criterion}
        v_c = 1 - \frac{1}{\rho(\widetilde G^{-1})\rho(B)}
    \end{equation}
    % then $v_c$ 
    is a target vaccination criterion.
\end{theorem}
\begin{proof}
    Recall that under Assumption \ref{ass:op_graph_strong_connect}, $B$ is strongly connected.
    Since $B$ is non-negative and irreducible, $\rho(B) > 0$, by the Perron–Frobenius theorem \cite{bapat_1998_a}. % .
    Therefore, the inverse of $\rho(B) \in \mathbb{R}_{>0}$ exists and $1 - \frac{1}{\rho(\widetilde G^{-1})\rho(B)}$ is well-defined.
    
    % \hmph{
    After showing that $v_c$ is well-defined, we will show that for any arbitrary $v(t)$ satisfying the target vaccination criterion, the healthy state equilibrium is locally asymptotically stable 
    for the corresponding infection subsystem \eqref{mat:x}, consistent with Definition~\ref{def:target_vaccination_criterion}.
    % $v(t) \gg v_c1_n$ for all $t \geq T$ where $T \geq 0$,
    % From Definition~\ref{def:target_vaccination_criterion}, $v(t)$ satisfies the target vaccination criterion if 
    Consider an arbitrary $v(t)$ such that 
    % has $v_i(t) > v_c$ for all $i\in [n]$ and $t \geq T$ for some $T \geq 0$. %, i.e.:
    % Therefore, 
    for all $i\in [n]$:
    \begin{align*}
     v_i(t) &> 1-\frac{1}{\rho(\widetilde G^{-1})\rho(B)}.% \nonumber\\ 
    %  \forall i\in [n], t \geq t_c.
    % \end{equation}
    \end{align*}
    % % For all $t > t_c$, it follows that:
    Therefore, for all $i\in [n]$,
    \begin{align}
       \left(\rho(\widetilde G^{-1})\rho(B)\right)^{-1} &> 1 - v_i(t). \label{equ:prf_spec_radius}
    \end{align}
    Since \eqref{equ:prf_spec_radius} holds for all $i\in [n]$,
    % $\left(\rho(\widetilde G^{-1})\rho(B)\right)^{-1} \geq \max\{1 - v_i(t): i\in [n]\}$,
    \begin{align*}\label{equ:x_stable_suff}
    \left(\rho(\widetilde G^{-1})\rho(B)\right)^{-1} &> \max\{1 - v_i(t): i\in [n]\} \\
        \Rightarrow \left(\rho(\widetilde G^{-1})\rho(B)\right)^{-1} &> \rho(I_n - \widetilde V(t)).
    \end{align*}
    By Lemma \ref{lemma:const}, $\widetilde {S}(t) \leq I_n - \widetilde {V}(t)$. 
    Therefore,  
    \begin{equation*}
        1 > \rho(\widetilde G^{-1})\rho(\widetilde{S}(t))\rho(B).
    \end{equation*}
    By the log-majorization theorem \cite[Thm 5.12]{ando1994majorizations}, we have $\sigma(A)\sigma(B) \geq \sigma(AB)$, where $\sigma(M)$ is the singular value of $M$. 
    Also note that $\rho(A) \leq \sigma(A)$ and equality holds if  $A$ is symmetric.
    Therefore, 
    since $\widetilde G^{-1}$ and $\widetilde S^*$ are diagonal and hence symmetric, 
    $\rho(\widetilde G^{-1})\rho(\widetilde S(t)) = \sigma(\widetilde G^{-1})\sigma(\widetilde S(t)) \geq \sigma(\widetilde G^{-1}\widetilde S(t)) = \rho(\widetilde G^{-1}\widetilde S(t))$.
    Similarly, since $B$ is symmetric under Assumption~\ref{asm:symm_B},
    \begin{align*}
        \rho(\widetilde G^{-1}S(t))\rho(B) &= \sigma(\widetilde G^{-1}S(t))\sigma(B)\\
        & \geq \sigma(\widetilde G^{-1}\widetilde S(t)B)\\
        & \geq \rho(\widetilde G^{-1}\widetilde S(t)B),
    \end{align*}
    leading to the following conclusion:
    \begin{equation}\label{equ:x_stable}
        1 > \rho(\widetilde G^{-1}\widetilde S(t) B),
        % \ \ \ \forall t \geq T\ \  \exists T \geq 0,
    \end{equation}
    which is the effective reproduction number of the infection subsystem for all time $t \geq T$.
It follows from Lemma~\ref{lem:x_rep_num} that the sub-system \eqref{mat:x} is locally asymptotically stable around $x^* = 0_n$ if \eqref{equ:x_stable} holds for all $t \geq T$ for some $T \in \mathbb{R}_{\geq 0}$.
    Therefore, $v_c$ is a target vaccination criterion.
    % by Definition~\ref{def:target_vaccination_criterion}.
    % \qed
\end{proof}
% Note that the eradication solution is trivial if $v_c \leq 0$ and infeasible if $v_c = 1$.
% Therefore, we assume that $v_c \in (0, 1)$.
% Furthermore, we define 
% \begin{equation*}
%     v_c = 1 - \frac{1}{\rho(\widetilde G^{-1})\rho(B)}
% \end{equation*}
% as the \textit{critical vaccination threshold} of system \eqref{mat:tot}.
Note that in the single-node case ($n=1$), $v_c = 1 - \frac{1}{\mathcal{R}_0}$, which is consistent with the literature \cite{keeling2008modeling}.

% The following lemma characterizes the critical opinion threshold of \eqref{mat:tot}. 
In the next lemma, we introduce the target opinion criterion $\underline o^{v^*}$, corresponding to some target vaccination goal $v^*$ and its corresponding control signal $\underline m^{v^*}$.
Notice that the target vaccination goal $v^*$ can be any constant such that $v_c \leq v^* < 1$.
% and $\underline o^{v^*}, \underline m^{v^*}$ are not functions of time, but functions of scalar values such that $\underline o: (0, 1) \to (0, 1)^n$, $\underline m: (0, 1) \to [0, \infty)^n$.
%%%%%
%%%%%
%%%%%
\begin{lemma}\label{lem:healthy_state_form}
If 
\begin{equation}\label{eq:target_opinion_tracking}
    \underline m^{v^*} = (I - \widetilde{\underline O}^{v^*})^{-1} (A + 2L[A])\underline o^{v^*},
\end{equation}
where:
\begin{equation}\label{eq:target_opinion_criterion}
    \underline o^{v^*} = k^{-1}[\Delta H](\frac{\delta}{1 - v^*} - k[H_{\min}]),
\end{equation}
then $(\underline o^{v^*}, (1 - v^*)1_n, 0_n, 0_n, v^* 1_n)$ 
is the unique healthy state equilibrium of system~\eqref{mat:tot} with the media actuator in \eqref{mat:o_w_input}.
\end{lemma}
\begin{proof}
The existence of the healthy state equilibrium can be shown by setting the state values in system~\eqref{mat:tot} with the media actuator in \eqref{mat:o_w_input} equal to $(\underline o^{v^*}, (1 - v^*)1_n, 0_n, 0_n, v^* 1_n)$, $m(t) = \underline m^{v^*}$ and observing that $\dot o = \dot s = \dot x = \dot r = \dot v = 0_n$.
The uniqueness of the healthy state equilibrium can be shown by assuming, 
% for 
by way of 
contradiction, that there exists a $p^* > \underline o^{v^*}$ such that:
\begin{equation}\label{eq:op_at_equil}
    0_n = A(0_n - p^*) - 2L[A]p^* + \widetilde{\underline M}^{v^*}.
\end{equation}
Moving terms around in \eqref{eq:op_at_equil} and using the definition of $\widetilde{\underline M}^{v^*}$ in \eqref{eq:target_opinion_tracking}, we obtain the following:
\begingroup\makeatletter\def\f@size{9}\check@mathfonts
\begin{equation*}
    (A + 2L[A])p^* = (I - \widetilde P^*)(I - \widetilde{\underline O}^{v^*})^{-1}(A + 2L[A])\underline o^{v^*}.
\end{equation*}
\endgroup
Since $p^* > \underline o^{v^*}$, which implies $(I - \widetilde P^*)(I - \widetilde{\underline O}^{v^*})^{-1} < I$, we obtain $(A + 2L[A])p^* < (A + 2L[A])\underline o^{v^*}$ by replacing $(I - \widetilde P^*)(I - \widetilde{\underline O}^{v^*})^{-1}$ with $I$.
Furthermore, since $A + 2L[A]$ is non-singular by the Gershgorin circle theorem, we 
multiply
% transform 
both sides by $(A + 2L[A])^{-1}$ and conclude that $p^* < \underline o^{v^*}$, which is a contradiction.
Therefore, $p^* > \underline o^{v^*}$ does not hold.
Without loss of generality, the opposite direction where $p^* < \underline o^{v^*}$ can also be rejected with a similar argument, leading to the conclusion that $p^* = \underline o^{v^*}$.
Therefore, since $\underline o^{v^*}$ is the unique solution of \eqref{mat:o_w_input} when $\dot o = 0_n$ and $\widetilde M(t) = \widetilde{\underline M}^{v^*}$,
and $H(o)$ is irreducible by definition in Section~\ref{sec:model} and Assumption~\ref{ass:op_graph_strong_connect} for all $o \in [0, 1]^n$,
we can replace $H_\min$ with $H(\underline o^{v^*})$ in \eqref{eq:healthy_sdot}, and then
it follows from the same argument in the proof of Proposition~\ref{Th:healthy_state_ex_n_uniq} that 
% $v^*$ is unique.
$(\underline o^{v^*}, (1 - v^*)1_n, 0_n, 0_n, v^* 1_n)$ 
is the unique healthy state equilibrium of system~\eqref{mat:tot}.
\end{proof}
Note that Lemma~\ref{lem:healthy_state_form} does not specify that the value of $v^*$ needs to be greater than $v_c$.
Instead, the lemma introduces a tracking scheme for any vaccination goal $v^* \in [0, 1]$.
A natural consequence of this lemma is that the system will return to its original equilibrium when the control signal is relaxed. 
% However,
Furthermore, we will see that not all vaccination goals are feasible via
% through 
opinion control.
The feasibility of the control signal depends on the upper and lower bounds of the opinions. that is, $\underline o_i^{v^*} \in [0, 1]$ $\forall i\in [n]$.
The next lemma characterizes the constraints in terms of the system parameters and the vaccination goal $v^*$.
\begin{lemma}\label{lem:o_feasiblity}
We have $0_n \leq \underline o_i^{v^*} \leq 1_n$ if and only if $(1 - v^*) k_i(H_{\min}) \leq \delta_i \leq (1 - v^*) k_i(H)$.
\end{lemma}
\begin{proof}
The following inequalities are equivalent by \eqref{eq:target_opinion_criterion}:
\begin{align*}
    0 &\leq \underline o_i^{v^*} \leq 1\\
    0 &\leq k_i^{-1}[\Delta H]\left(\frac{\delta_i}{1 - v^*} - k_i[H_{\min}]\right) \leq 1\\
    0 &\leq  \frac{\delta_i}{1 - v^*} - k_i[H_{\min}] \leq k_i[\Delta H]\\
    % k_i[H_{\min}] &\leq  \frac{\delta_i}{1 - v^*} \leq k_i[H]\\
    (1 - v^*)k_i[H_{\min}] &\leq  \delta_i \leq (1 - v^*)k_i[H],
\end{align*}
which concludes the proof.
\end{proof}
Lemma~\ref{lem:o_feasiblity} shows that only a subset of systems of the form in~\eqref{mat:tot} 
% parameterizations 
% admits meaningful control signals through opinion
can be controlled via opinions, that is, there exist parameterizations of \eqref{mat:tot} such that the vaccination goal $v^* \geq v_c$ can be achieved only if $\underline o_i^{v^*} > 1$ or $\underline o_i^{v^*} < 0$.
While there is flexibility in choosing $v^*$, the fundamental constraint lies in the relationship between the rate of loss of immunity from vaccination $\delta_i$, the target vaccination criterion $v_c$, and the maximum rate of vaccine-uptake $k_i[H]$.
If the ratio between the rate of loss of immunity from vaccination and the maximum rate of vaccine-uptake is greater than $1 - v_c$, then no control signals of the form in \eqref{mat:o_w_input} are guaranteed to eradicate the disease, unless we have a more relaxed $v_c$ that guarantees stability around the healthy state equilibrium of the system.

The next theorem shows that the suggested control law %provides
can provide local stability around the healthy state equilibrium of system~\eqref{mat:tot} with a media actuator \eqref{mat:o_w_input}.
\begin{theorem} \label{Thm:control_with_actu}
If $0_n \ll \underline o^{v^*} \ll 1_n$ and $m(t) = \underline m^{v^*}$ as defined in \eqref{eq:target_opinion_criterion} and \eqref{eq:target_opinion_tracking}, respectively, then for all $v^* \geq v_c$ 
% system~
\eqref{mat:tot} with media actuator \eqref{mat:o_w_input} is locally asymptotically stable around the healthy state equilibrium $(\underline o^{v^*}, (1 - v^*)1_n, 0_n, 0_n, v^* 1_n)$.
\end{theorem}

\begin{proof}
Since $s(t) = 1_n - x(t) - r(t) - v(t)$ by Lemma~\ref{lemma:const}, we can linearize 
% system~
%\eqref{mat:tot} with input as follows without loss of generality:
\eqref{mat:tot} around the healthy state equilibrium as follows, without including the dynamics of $s(t)$:
\begin{equation*}
    \begin{bmatrix}
        \dot x\\
        \dot o\\
        \dot r\\
        \dot v
    \end{bmatrix} = \begin{bmatrix}
        J_x & 0 & 0 & 0\\
        A & J_o & 0 & 0\\
        \widetilde G & 0 & J_r & 0\\
        0 & v^*(1 - v^*)I & v^*\widetilde K[H(\underline o)] & J_v
    \end{bmatrix}\begin{bmatrix}
        x\\
        o\\
        r\\
        v
    \end{bmatrix}
\end{equation*}
where $J_o = -(A + 2L[A] + \widetilde{\underline M}^{v^*})$, $J_x = (1 - v_c)B - \widetilde G$, $J_r = -(\widetilde W + v^* K[H(\underline o^{v^*})])$, and $J_v = (1 - v^*)H(\underline o^{v^*}) - \widetilde D$.

The linearized system is stable if and only if each diagonal block is stable because the Jacobian is lower triangular.
By Theorem~\ref{thm:sufficient_x_stable}, the infection subsystem $J_x$ is locally asymptotically stable around $x^* = 0_n$ for all $v^* \geq v_c$.
Moreover, we notice that each Gershgorin disc of $J_o$ has a center at $-(\underline m_i^{v^*} + 2k_i[A])$ with radius $k_i[A]$, and $\underline m_i^{v^*}$ is non-negative.
Therefore, the spectral abscissa of $J_o$ lies on the open left half plane by the Gershgorin circle theorem.
Furthermore, $J_r$ is locally asymptotically stable by observing that %$-(\widetilde W + v_c K[H(\underline o)])$ is a negative diagonal
$-(\widetilde W + v^* K[H(\underline o^{v^*})])$ is a negative diagonal
matrix.

Lastly, an upper bound of the spectral abscissa of $J_v$ can be obtained by replacing $\widetilde D$ with $(1 - v^*)k_i[H_{\min}]$,
%It is because by assumption, $\underline o \gg 0_n$, and it follows from Lemma~\ref{lem:o_feasiblity} that $\widetilde D \gg (1 - v^*) \widetilde K(H_{\min})$ if $\underline o \gg 0_n$. 
because Lemma~\ref{lem:o_feasiblity} states that $\widetilde D \gg (1 - v^*) \widetilde K(H_{\min})$ if $\underline o^{v^*} \gg 0_n$. By assumption, $\underline o^{v^*} \gg 0_n$ holds.
Then, by shifting the center of each Gershgorin disc to the left, we obtain:
\begin{align}\label{eq:pf_ctrl_with_actu_J_v}
    \alpha(J_v) &= \alpha((1 - v^*)H(\underline o^{v^*}) - \widetilde D)\nonumber\\
    &< \alpha((1 - v^*)H(\underline o^{v^*}) - (1 - v^*) \widetilde K[H_{\min}])\nonumber\\
    & = \alpha \left((1 - v^*)\left[H(\underline o^{v^*}) - \widetilde K[H_{\min}]\right]\right).
\end{align}
% If we examine $H(\underline o)$ b
By replacing each term 
% in $H(\underline o)$ 
with 
% their
its
corresponding definition: $H(\underline o^{v^*}) := \widetilde{\underline O}^{v^*}\Delta H + (H_{\min} - K[H_{\min}])$ and $\widetilde{\underline O}^{v^*} := \widetilde K^{-1}[\Delta H](\frac{1}{1 - v^*}\widetilde D - \widetilde K[H_{\min}])$, 
% then 
we have the following:
\begingroup\makeatletter\def\f@size{9}\check@mathfonts
\begin{equation}\label{eq:pf_ctrl_with_actu_expand}
    H(\underline o^{v^*})
    % % &= \widetilde{\underline O}\Delta H + (H_{\min} - K[H_{\min}])\\
     = \left(\frac{1}{1 - v^*} \widetilde D - \widetilde K[H_\min] \right)\widetilde K^{-1}[\Delta H] + H_{\min}.
\end{equation}
\endgroup
Notice that by replacing $\widetilde D$ with $(1 - v^*) \widetilde K(H_{\min})$ in \eqref{eq:pf_ctrl_with_actu_expand}, terms cancel out and we are left with $H(\underline o^{v^*}) = H_{\min}$. 
% Therefore, if we apply the inequality $\widetilde D \gg (1 - v^*) \widetilde K(H_{\min})$ to \eqref{eq:pf_ctrl_with_actu_J_v} again, we have:
% Therefore, if we apply the inequality $\widetilde D \gg (1 - v^*) \widetilde K(H_{\min})$ to \eqref{eq:pf_ctrl_with_actu_expand} and substitute $H(\underline o)$ with $H_{\min}$ in \eqref{eq:pf_ctrl_with_actu_J_v}, we have:
Therefore, if we substitute $H(\underline o^{v^*})$ from \eqref{eq:pf_ctrl_with_actu_expand} into \eqref{eq:pf_ctrl_with_actu_J_v} 
% with \eqref{eq:pf_ctrl_with_actu_expand} 
and apply the inequality $\widetilde D \gg (1 - v^*)$, we have:
\begin{align*}
    \alpha(J_v) &< \alpha \left((1 - v^*)\left[H(\underline o^{v^*}) - \widetilde K[H_{\min}]\right]\right)\\
    & < \alpha \left((1 - v^*)\left[H_{\min} - \widetilde K[H_{\min}]\right]\right).
\end{align*}
By the Gershgorin circle theorem, $\alpha \left((1 - v^*)\left[H_{\min} - \widetilde K[H_{\min}]\right]\right) \leq 0$, which completes the proof.
\end{proof}

\section{Simulation}\label{sec:sim}
In this section, we illustrate and visualize the results in Section~\ref{sec:main}. 
% We suppose 
% There are 4 nodes (communities) in our simulated system. 
There are 4 nodes in our simulated system, where each node represents an individual or a community.
The key difference between a network of communities and a network of individuals in a networked epidemic model is that whether we allow self-loops in the disease transmission network.
When modelling interaction between individuals, it is most natural to assume that each individual does not induce sickness to themselves, while the opposite is true when modelling interaction between communities.

% We will show that when the epidemic reaches the healthy state equilibrium, the opinions will reach consensus at $o^* = 0_n$ (Lemma \ref{lem:healthycons}).
%\Cref{fig:H} is graphical representation of networks used in the simulation. 
% Let $a_{ij} = 0.1$, $\forall i,j \in [n]$,
The infectious network, the strategy immitation network, and the opinion network are characterized by the following matrices:
\begingroup\makeatletter\def\f@size{7}\check@mathfonts
\begin{align*}
    B &= \left[\begin{matrix}
        0& 0.4& 0.35 & 0.3  \\
        0.4& 0&  0.2& 0.25  \\
        0.35& 0.2& 0& 0.2   \\
        0.3&  0.25& 0.2 & 0 \\
    \end{matrix}\right],~
    % &H = \left[\begin{matrix}
    %     0.&   1.0&   1.   &0.5 \\
    %     0.7&  0.&    0.8   &1.0 \\
    %     1.0&  0.5&   0.   &1.0 \\
    %     1.0&  1.0&   0.5   &0. \\
    H = \left[\begin{matrix}
        0&   0.4&  &0.2   &0.5 \\
        0.7&  0&   &0.8   &1 \\
        1&  0.5&   &0    &0.7\\
        1&  0.8&   &0.5   &0 \\
    \end{matrix}\right].%,\\
    % H_{\min} &= \left[\begin{matrix}
    %     0.0&   0.02&   0.02   &0.02 \\
    %     0.03&  0.0&    0.03   &0.03 \\
    %     0.04&  0.04&   0.0   &0.04 \\
    %     0.01&  0.01&   0.01   &0.0 \\
    % \end{matrix}\right],
    % &A &= \left[\begin{matrix}
    %     0.1&   0.1&   0.1   &0.1 \\
    %     0.1&  0.1&    0.1   &0.1 \\
    %     0.1&  0.1&   0.1   &0.1 \\
    %     0.1&  0.1&   0.1   &0.1\\
    % \end{matrix}\right].
\end{align*}
\endgroup
Other parameters are defined as follows: 

\begin{itemize}
    \item $\gamma = [0.4,1.6,0.8,0.3]$;
    \item $\omega = [0.4,0.5,0.6,0.7]$;
    \item $\delta = [0.1,0.2,0.2,0.1]$;
    \item $\eta^{\min}_{ij} = 0.1$, $\forall i\neq j \in [n]$;% and $i \neq j$;
    \item $\alpha_{ij} = 0.1$, $\forall i,j \in [n]$.
\end{itemize}
% $\gamma = [0.6,1.6,0.8,0.3]$, 
% $\omega = [0.4,0.5,0.6,0.7]$, 
% $H_{\min} =0.1$, 
% $\rho = [0.2,0.3,0.4,0.1]$, and 
% $\delta = [0.1,0.2,0.2,0.1]$.
%For each community $i$, 
The initial conditions of the system:
$$(s(t_0),x(t_0),r(t_0),v(t_0),o(t_0)),$$
are
$((0,0.2, 0.3, 0.2)$, $(0.2, 0.4, 0.2, 0.3)$,
$(0.6, 0.3, 0.5, 0.5)$, $(0.2,0.1,0,0)$, $(0,0.1,0.9,1))$, respectively, where $t_0 = 0$ in the following simulations. 

%Fig.~\ref{fig:system} shows the dynamics of the system. As {indicated} by Proposition \ref{prop:x_stability}, the system ends up in the the healthy state (Fig. \ref{fig:i}), where $\sigma(\widetilde{S}^* B-\widetilde G) = -0.1413$. % and $\sigma(-\tilde{S}^*\tilde PH^*-\tilde D) = 0.00571$.
%It
%The opinions of the system also reaches  consensus (Fig. \ref{fig:o}) as stated in Lemma \ref{lem:healthycons}.
% \begin{figure}[t]
% \centering
%         \begin{subfigure}[b]{0.23\textwidth}
%          \centering
%          \includegraphics[width=\textwidth]{S.png}
%         \caption{}\label{fig:s}
%         \end{subfigure}
%      \hfil
%      \begin{subfigure}[b]{0.23\textwidth}
%          \centering
%          \includegraphics[width=\textwidth]{i.png}
%          \caption{}\label{fig:i}
%      \end{subfigure}

%     % \begin{subfigure}[b]{0.23\textwidth}
%     %     \centering
%     %     \includegraphics[width=\textwidth]{r.png}
%     %    \caption{} \label{fig:r}
%     % \end{subfigure}
%      \hfil
%      \begin{subfigure}[b]{0.23\textwidth}
%          \centering
%          \includegraphics[width=\textwidth]{V.png}
%         \caption{} \label{fig:v}
%      \end{subfigure}
%      \hfil
%     \begin{subfigure}[b]{0.23\textwidth}
%          \centering
%          \includegraphics[width=\textwidth]{o.png}
%          \caption{} \label{fig:o}
%      \end{subfigure}

%      \caption{Simulated system dynamics.}
%      \label{fig:system}
% \end{figure}

% \par To show Theorem \ref{Th:healthy_state_ex_n_uniq}, 
Figure~\ref{fig:uniq} illustrates the existence 
% and uniqueness 
of the healthy state equilibrium with a non-zero vaccination equilibrium which is formally presented in Proposition~\ref{Th:healthy_state_ex_n_uniq}.
We change the initial conditions (Figure~\ref{fig:uniq} (b)) to:
            $s(t_0) = (0.9,0.2,0.5,0.4)$,
            $x(t_0) = (0,0.4,0.2,0.1)$,
            $r(t_0) = (0.1,0,0,0)$,
            $v(t_0) = (0,0.4,0.3,0.5)$, and
            $o(t_0)$ $=$ $(0,0.1,$ $0.9,1)$.
Figure~\ref{fig:uniq} shows that $v^* \gg 0$ is invariant to a single perturbation in the initial states. 
According to Lemma \ref{lemma:healthy}, at a healthy equilibrium, if $v^*$ is unique, $s^*$, $r^*$, and $x^*$ are also unique. 
% in the healthy state equilibrium.
\begin{figure}
    \centering
\begin{subfigure}{.23\textwidth}
  \centering
  \includegraphics[width=1\linewidth]{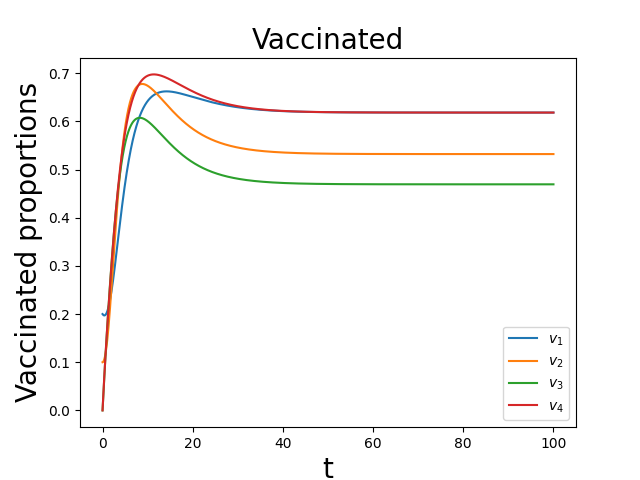}
  \caption{$v(t)$ with the original \\initial condition.}
  \label{fig:sub_orig}
\end{subfigure}%
\begin{subfigure}{.23\textwidth}
  \centering
  \includegraphics[width=1\linewidth]{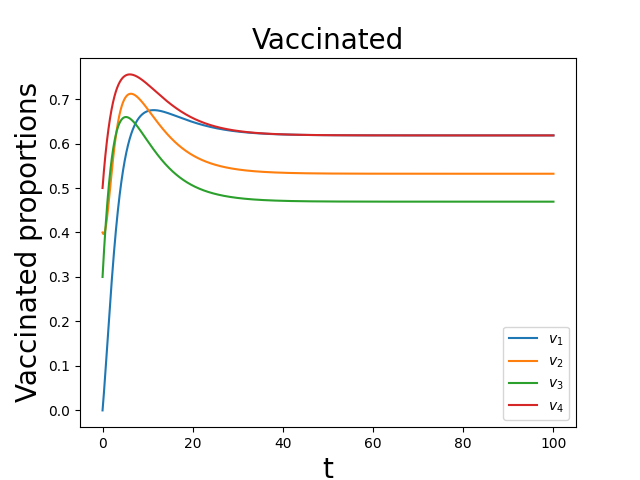}
  \caption{$v(t)$ with the second \\initial condition.}
  \label{fig:sub_init}
\end{subfigure}
    \caption{The existence and uniqueness of the healthy state equilibrium, illustrating Proposition~\ref{Th:healthy_state_ex_n_uniq}. Only the vaccination levels over time are shown.}
    \label{fig:uniq}
\end{figure}

\par To show the impact of the proposed 
% state feedback 
controller, we first construct a system with an endemic state equilibrium, which is as an equilibrium where $x^* > 0_n$,
% , where $x^* \neq 0_n$, 
by setting
$\gamma = [0.2,0.8,0.4,0.15]$.
% \begingroup\makeatletter\def\f@size{7}\check@mathfonts
% \begin{equation*}
% H = \left[\begin{matrix}
%         0  & 1.3 & 1.3 & 0.65\\
%         1.12& 0  & 1.28& 1.6 \\
%         1.4 & 0.7 & 0  & 1.4 \\
%         2  & 2  & 1  & 0   \end{matrix} \right], \text{ and }
%         H_{min} = \left[\begin{matrix}
%         0& 0.13  &0.13 &0.13 \\
%         0.16& 0 &0.16  &0.16 \\
%         0.14& 0.14 &0  &0.14\\
%         0.15& 0.15 &0.15 &0 \end{matrix} \right].
% \end{equation*}
% \endgroup
Other parameters are the same as in the previous setting. 
Note $B$ is symmetric in the original setup. 
Figure~\ref{before_control} shows the dynamics of the system without the control strategy.
\begin{figure}
    \centering
\begin{subfigure}{.23\textwidth}
  \centering
  \includegraphics[width=1\linewidth]{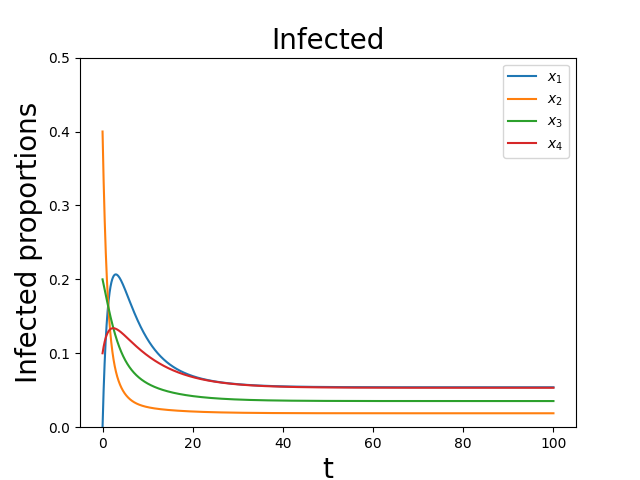}
  %\caption{}
  \label{fig:sub_i_before}
\end{subfigure}%
% \begin{subfigure}{.23\textwidth}
%   \centering
%   \includegraphics[width=1\linewidth]{opinion_before_control.png}
%   \caption{}
%   \label{fig:sub_o_before}
% \end{subfigure}
\begin{subfigure}{.23\textwidth}
  \centering
  \includegraphics[width=1\linewidth]{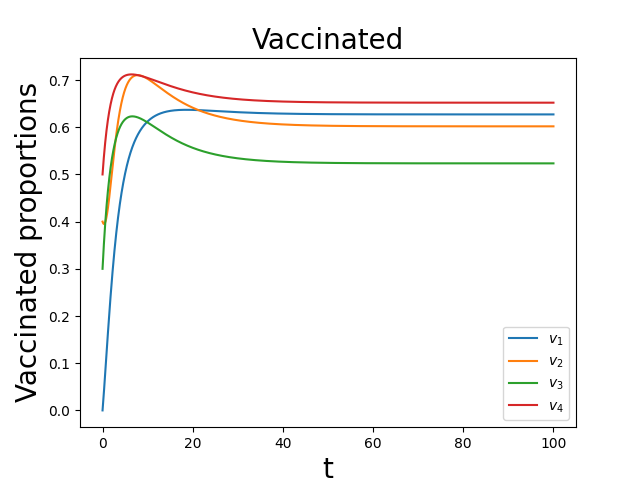}
  %\caption{}
  \label{fig:sub_v_before}
\end{subfigure}
    \caption{The system with an endemic state equilibrium.}
    \label{before_control}
\end{figure}
% \par Then, we apply the control strategy based on Algorithm~\ref{alg:cap}. 
% We choose $m_i = 1$. If there does not exist a $o_i(t_1)$ that satisfies the condition in Lemma \ref{lemma:mono_v_p1}, $m_i (t) = m_i = 1$, $\forall t > t_1$, and $m_i(t) = 0$, $\forall t <t_1$. 
% In Figure~\ref{after_control}, the system with the applied control strategy results in a healthy state equilibrium. Note we keep the control on to make the healthy state equilibrium stable.
\par Then, we apply the control strategy based on Theorem~\ref{Thm:control_with_actu}. 
If there exists a $t_1^i$ such that $o_i(t_1^i)< \underline o_i^{v_c}$ defined in \eqref{eq:target_opinion_criterion}, then $m_i (t) = \underline m_i^{v_c}$ defined in \eqref{eq:target_opinion_tracking}, $\forall t > t_1^i$, and $m_i(t) = 0$, $\forall t < t_1^i$, where $t_1^i$ represents when the control applied on Node $i$ starts. 
In our case, as shown by Figure \ref{after_control}, $t_1^{[1-3]}$ are $0$, $0$, $2.29$, respectively. 
Node 4 is never connected to the actuator, since $o_4(t) \geq \underline o_i^{v_c}$, $\forall t > t_0$. 
In Figure~\ref{after_control}, we can see that the system with the applied control strategy results in a healthy state equilibrium. 
Note, we need to 
% keep 
maintain
the control input in order to ensure that
% on to make 
the healthy state equilibrium remains stable. 
Figure~\ref{end_control} shows that the healthy state equilibrium is no longer stable if we remove the media actuator at $t_2 = 74$, and the system returns to its natural endemic equilibrium. As indicated by the red dashes in Figure~\ref{end_control}, the system no longer maintains the healthy state equilibrium $\forall t > t_2$. 
% that $\lim_{\epsilon \to 0_n}x(t_2) < \epsilon$.

\begin{figure}
    \centering
\begin{subfigure}{.23\textwidth}
  \centering
  \includegraphics[width=1\linewidth]{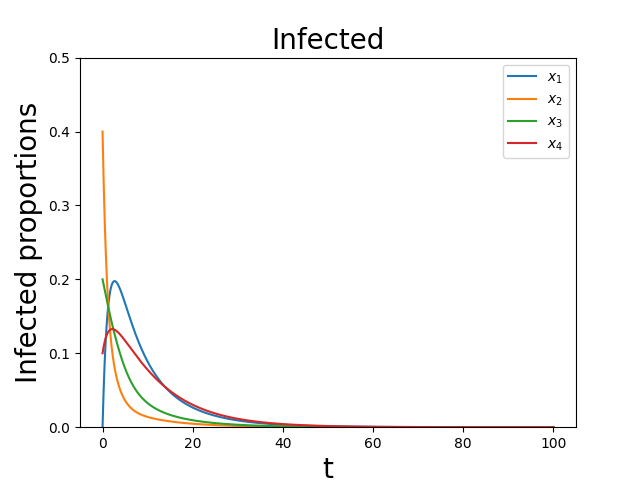}
%   \caption{}
  \label{fig:sub_i_after}
\end{subfigure}%
\begin{subfigure}{.23\textwidth}
  \centering
  \includegraphics[width=1\linewidth]{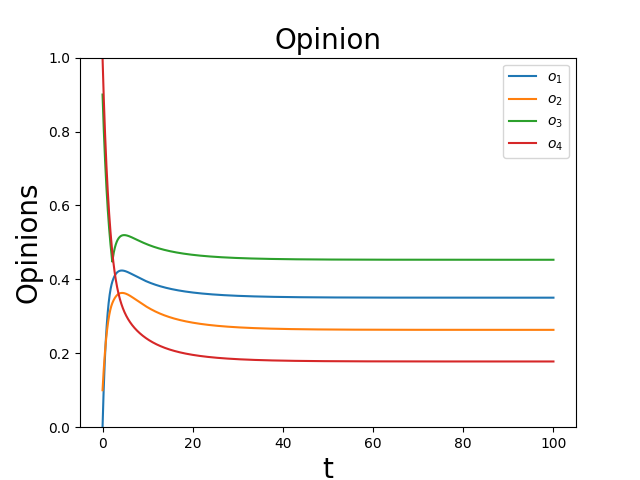}
%   \caption{}
  \label{fig:sub_o_after}
\end{subfigure}
\begin{subfigure}{.23\textwidth}
  \centering
  \includegraphics[width=1\linewidth]{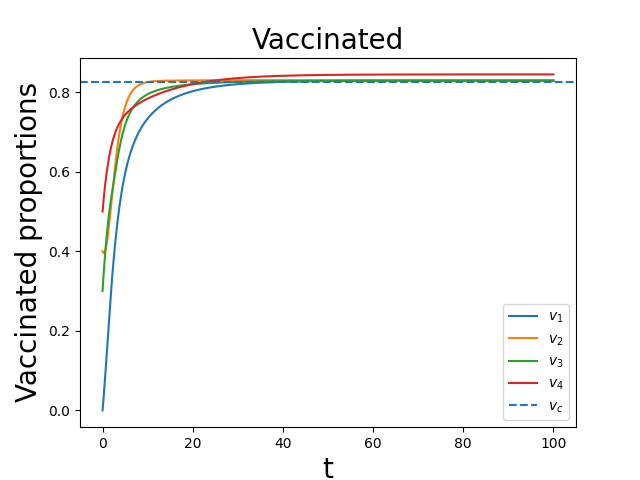}
%   \caption{}
  \label{fig:sub_v_after}
\end{subfigure}
\begin{subfigure}{.23\textwidth}
  \centering
  \includegraphics[width=1\linewidth]{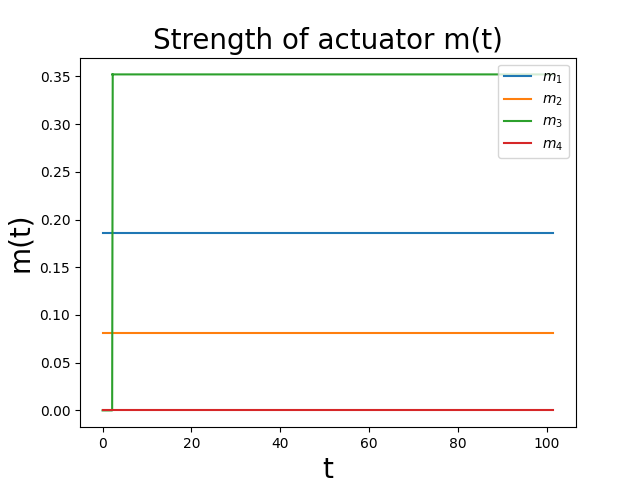}
%   \caption{}
  \label{fig:strength}
\end{subfigure}
    \caption{The system with the control strategy applied.}
    \label{after_control}
\end{figure}

\begin{figure}
    \centering
\begin{subfigure}{.23\textwidth}
  \centering
  \includegraphics[width=1\linewidth]{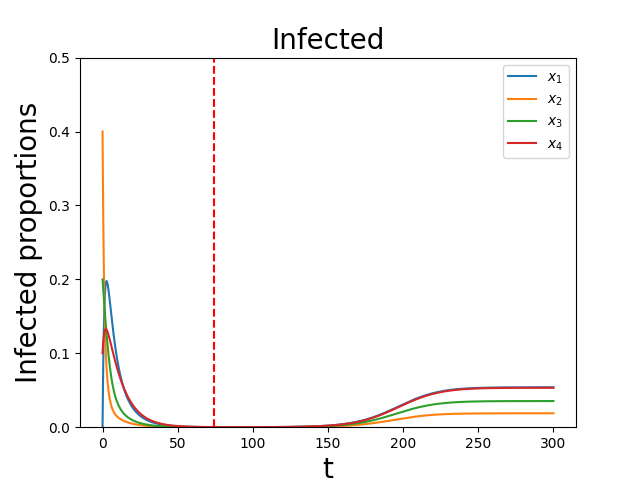}
%   \caption{}
  \label{fig:sub_i_shutdown}
\end{subfigure}%
\begin{subfigure}{.23\textwidth}
  \centering
  \includegraphics[width=1\linewidth]{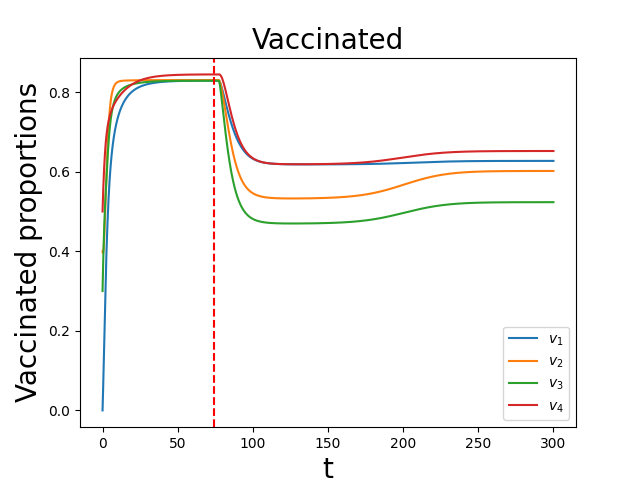}
%   \caption{}
  \label{fig:sub_v_shutdown}
\end{subfigure}
    \caption{The healthy state equilibrium is unstable after ending the control. The vertical red dotted line indicates the time when the control signal is lifted.}
    \label{end_control}
\end{figure}

\section{Conclusion}\label{sec:conc}
In this work, we propose a networked $SIRS-V_o$ model to study the interaction between opinion dynamics, self-interest decision-making, vaccine uptake,
% program 
and disease spreading process. 
We provide in this paper a proof on the existence of the healthy state equilibrium of the networked $SIRS-V_o$ model, and show that the problem can be reduced to the equilibrium analysis of a single viral networked SIS model.
We then characterize 
% the basic and effective reproduction number 
the local stability condition of the infection subsystem of the networked multi-layer epidemic model around the healthy state equilibria.
Building on this model, we develop target vaccination and opinion criteria %on 
which guarantee %s 
disease eradication.
These criteria provide sufficient conditions that leverage opinions and vaccination to eradicate the spread of an epidemic.
We then show analytically that the family of control signals which satisfies the target vaccination and opinion criteria stabilizes the networked $SIRS-V_o$ model locally around the healthy state equilibrium under some mild assumptions.
% sufficient condition for asymptotic convergence of the infection subsystem around the healthy state equilibrium under some mild assumptions.
Among these mild assumptions, the requirement that the target vaccination criterion is feasible at node $i$ only when the ratio between the rate of loss of immunity $\delta_i$
% from vaccination 
and the maximum rate of vaccine-uptake $k_i[H]$ is greater than $1 - v_c$ gives us insight into the limitation of epidemic control through the influence of a centralized media actuator.
We also note that the infection subsystem returns to its natural equilibrium after the control algorithm is relaxed, both analytically and in simulations.
% Besides the development and analysis on the state feedback controller, we provide in this paper a proof on the existence and uniqueness of the healthy state equilibrium without the actuator, and we characterize the basic and effective reproduction number of the networked multi-layer epidemic model.
In the future, we will extend the control analysis to study robust government-subsidized vaccination programs for disease control under adversarial attacks on opinions through fake news or other mediums.

% In this work, we develop a networked $SIRS-V_o$ model to study the dynamic interaction of opinion, vaccination, and epidemic spreading. 
% The opinion dynamics in the model represents the interchange of people's attitudes toward the severity of the disease. 
% The model includes the perceived risk network that determines the vaccination intention of the communities. 
% For the opinion network, the system will reach consensus at $0$ in the healthy state equilibrium. 
% % The novelty of this work is coupling the perceived risk with the opinion dynamics to represent people's demand or motivation of vaccination. 
% % Even if they think the epidemic disease is no longer serious, they still perceive a minimum threat from the neighborhood, which keeps them taking the vaccine.
% Then, we analyze the local stability conditions of our model and prove that the healthy state equilibrium is unique. 
% Finally, we study the conditions and factors that determine the state of the system. 
% With those results, we propose a potential control strategy, leveraging vaccination and opinions, to eradicate the epidemic spreading. 
% The essence of the strategy is to promote people's awareness; it indirectly increases the vaccinated population; then the susceptible population decreases; and the strategy finally eradicates the epidemics. 
% The algorithm can ensure the healthy state equilibrium or draw an endemic state to the healthy one.

% \section{Proof of Theorem~\ref{thm:observer}}
% \label{append:observer}

\bibliographystyle{IEEEtran}
\bibliography{refs}

\end{document}